\pdfoutput=1
\documentclass[a4paper,USenglish,numberwithinsect,thm-restate,nolineno]{socg-lipics-v2021}
\hideLIPIcs

\usepackage{hyperref}
\hypersetup{
    colorlinks=true,
    linkcolor=blue
    }

\usepackage{microtype}%

\usepackage{array}
\usepackage{cite} 
\usepackage{mathtools}                
\usepackage[mathlines]{lineno}
\usepackage{amsfonts}
\usepackage{amsthm}
\usepackage{xcolor}
\usepackage{algorithm}
\usepackage{algpseudocode}
\theoremstyle{remark}

\newtheorem*{problem*}{Problem}
\newtheorem*{theorem*}{Theorem}

\DeclareMathOperator{\poly}{\mathrm{poly}}
\DeclareMathOperator{\polylog}{\mathrm{polylog}}

\let\eps\varepsilon
\def\reals{\mathbb{R}}

\renewcommand{\P}{\mathcal{P}}
\DeclareMathOperator{\disk}{\mathrm{disk}}

\newcommand{\old}[1]{{}}

\let\eps\varepsilon

\newcommand{\ddF}{d_{\textnormal{dF}}}
\newcommand{\dF}{d_{\textnormal{dF}}}
\newcommand{\Tfvd}{T_{\textsc{fvd}}}
\newcommand{\Tvd}{T_{\textsc{vd}}}
\newcommand{\Tds}{T_{\textsc{ds}}}
\newcommand{\Tannu}{T_{\textsc{annu}}}
\newcommand{\Tedge}{T_{\textsc{edge}}}

\usepackage{xcolor}
\newcommand{\boo}[3]{\textcolor{#1}{#2: #3}}

\newcommand{\indu}{\boo{red}{Indu}}

\EventEditors{Wolfgang Mulzer and Jeff M. Phillips}
\EventNoEds{2}
\EventLongTitle{40th International Symposium on Computational Geometry
(SoCG 2024)}
\EventShortTitle{SoCG 2024}
\EventAcronym{SoCG}
\EventYear{2024}
\EventDate{June 11-14, 2024}
\EventLocation{Athens, Greece}
\EventLogo{socg-logo}
\SeriesVolume{293}
\ArticleNo{XX}     %

\title{%
  Discrete Fr\'echet Distance Oracles\thanks{A preliminary version of this work will appear in \emph{SoCG'24} \cite{full-SoCG-no-arXiv}.}
}

\author{Boris Aronov}{Department of Computer Science and Engineering, Tandon School of Engineering, New York University, Brooklyn, NY 11201 USA.}
{boris.aronov@nyu.edu}
{https://orcid.org/0000-0003-3110-4702}
{Work partially supported by NSF Grant CCF-20-08551.  Part of the work was done while visiting Institute of Science and Technology Austria.}

\author{Tsuri Farhana}{Department of Computer Science, Ben Gurion University, Beer Sheva, Israel.}{tsurif@post.bgu.ac.il}{}{}

\author{Matthew J. Katz}
{Department of Computer Science, Ben Gurion University, Beer Sheva, Israel.}
{matya@cs.bgu.ac.il}
{https://orcid.org/0000-0002-0672-729X}
{Work partially supported by Grant 2019715/CCF-20-08551 from the US-Israel Binational Science Foundation/US National Science Foundation.}

\author{Indu Ramesh}{Department of Computer Science and Engineering, Tandon School of Engineering, New York University, Brooklyn, NY 11201 USA.}{ir914@nyu.edu}{https://orcid.org/0009-0008-9967-0819}{Work supported by a Tandon School of Engineering Fellowship and by NSF Grant CCF-20-08551.}

\titlerunning{Discrete Fr\'echet Distance Oracles}

\authorrunning{B. Aronov, T. Farhana, M.J. Katz, and I. Ramesh}
\Copyright{Boris Aronov, Tsuri Farhana, Matthew J. Katz, and Indu Ramesh}



\acknowledgements{We would like to thank an anonymous reviewer of an earlier version of this paper for a suggested improvement.}

\begin{document}

\maketitle

\begin{abstract}
It is unlikely that the discrete Fr\'echet distance between two curves of length $n$ can be computed in strictly subquadratic time.
We thus consider the setting where one of the curves, $P$, is known in advance. 
In particular, we wish to construct data structures (\emph{distance oracles}) of near-linear size that support efficient distance queries with respect to $P$ in sublinear time. Since there is evidence that this is impossible for query curves of length $\Theta(n^\alpha)$, 
for any $\alpha > 0$, we focus on query curves of (small) constant length, for which we are able to devise distance oracles with the desired bounds.

We extend our tools to handle subcurves of the given curve, and even arbitrary vertex-to-vertex subcurves of a given geometric tree. That is, we construct an oracle that can quickly compute the distance between a short polygonal path (the query) and a path in the preprocessed tree between two query-specified vertices. Moreover, we define a new family of geometric graphs, \emph{$t$-local} graphs (which strictly contains the family of geometric spanners with constant stretch), for which a similar oracle exists: we can preprocess a graph $G$ in the family, so that, given a query segment and a pair $u,v$ of vertices in $G$, one can quickly compute the smallest discrete Fr\'echet distance between the segment and any $(u,v)$-path in $G$.  The answer is exact, if $t=1$, and approximate if $t>1$.
\end{abstract}

\section{Introduction}
\label{sec:intro}

The continuous Fr\'echet distance is often used as a measure of similarity between curves~\cite{AltG95}. The discrete Fr\'echet distance~\cite{EiterM94} is sometimes viewed as a simplified version of the (continuous) Fr\'echet distance, but it is also the preferred version in some application domains, such as protein alignment (see, e.g.,~\cite{WylieZ13}).

Let $A=(a_1,\ldots,a_m)$ and $B=(b_1,\ldots,b_n)$ be two sequences of points in $\reals^d$ representing polygonal curves. A \emph{(monotone) walk} of $A$ and $B$ is a sequence of pairs $(c_1,\ldots,c_l)$, where (i) $c_1=(a_1,b_1)$, (ii) $c_l=(a_m,b_n)$, and (iii) the pair succeeding $c_k=(a_i,b_j)$, for $1 \le k < l$, is one of the following: $c_{k+1}=(a_{i+1},b_j)$ (when $i < m$), $c_{k+1}=(a_i,b_{j+1})$ (when $j < n$), or $c_{k+1}=(a_{i+1},b_{j+1})$ (when $i < m$ and $j < n)$.
Each pair $c_k=(a_i,b_j)$ in a walk $(c_1,\ldots,c_l)$ of $A$ and $B$ yields a distance $\|a_i-b_j\|$, and the \emph{cost} of the walk is the maximum of these distances. The \emph{discrete Fr\'echet distance} between $A$ and $B$, denoted $\ddF(A,B)$, is the minimum over the cost of all walks of $A$ and $B$.

The discrete Fr\'echet distance between $A$ and $B$ can be computed in roughly $O(mn)$ time~\cite{AgarwalAKS14,EiterM94}. It is unlikely that it can be computed exactly, or even approximated within a factor less than 3, in strictly subquadratic time~\cite{Bringmann14,BringmannM16,BuchinOS19}. It is therefore natural to ask whether one can do better when, e.g., one of the curves is given in advance.
Indeed, let $G$ be a geometric graph, that is, $G$'s vertices correspond to points in the plane, and the weight of an edge of $G$ is the Euclidean distance between the points represented by its vertices. Denote the set of paths from $u$ to $v$ in $G$, where $u$ and $v$ are vertices of $G$, by~$\P_G(u,v)$. (If $G$ is a tree, then $\P_G(u,v) = \{\Pi_{uv}\}$, where $\Pi_{uv}$ is the unique path in $G$ from $u$ to $v$.)
The \emph{discrete Fr{\' e}chet distance between a polygonal curve $Q$ and $G$} (with respect to $u$ and $v$) is $\min_{\Pi \in \P_G(u,v)} \ddF(Q,\Pi)$, and we denote this distance by $\ddF(Q,\P_G(u,v))$.
Now, assume that we are expecting a stream of polygonal curves $Q_1,Q_2,\ldots$,
each with a corresponding 
pair $(u_i,v_i)$ of vertices of $G$, and for each arriving $Q_i$ we need to compute the distance $\ddF(Q_i,\P_G(u_i,v_i))$. 
We thus wish to construct a compact data structure based on $G$, so that given a query curve $Q$ and two vertices $u$ and $v$ of $G$, one can compute $\ddF(Q,\P_G(u,v))$ efficiently. In other words, we wish to construct a \emph{distance oracle} for $G$.

To construct such a data structure, we focus on the case where the curves $Q_i$ are of constant size, i.e., consist of a constant number of vertices (at the other extreme, when queries have size $\Theta(n^\alpha)$, for $\alpha>0$, it may be impossible to gain anything by polynomial-time preprocessing \cite{GudmundssonSW23, BringmannKN21}). In this case, the challenge is to construct a near-linear size data structure such that given a curve $Q$, one can compute $\ddF(Q,\P_G(u,v))$ in sublinear time. We identify several rather general settings where this is possible. Specifically, if $G$ is a tree with $n$ nodes, we can process query curves of size up to three in $O(\polylog n)$ time and curves of size four in $O^*(n^{1/2})$ time; the $O^*(\cdot)$ notation hides subpolynomial factors. (We get slightly better bounds for the special case where the tree is actually a polygonal curve.) Moreover, we define a class of geometric graphs, called \emph{1-local} graphs, which includes the Delaunay graph, for which we can answer segment queries in $O^*(n^{1/2})$ time.

\subparagraph*{Our results.}
We first formally state the main problem studied in this paper.

\begin{problem*}[Distance Oracle]
Let $G$ be a geometric graph.
Construct a compact data structure such that, given a query polygonal curve $Q$ of length (i.e., number of vertices) $k$ and two vertices $u$ and $v$ of $G$, one can quickly compute $\ddF(Q,\P_G(u,v))$.
\end{problem*}

We assume that the sets of points underlying $G$ and $Q$ are in 
the plane, and we focus on the case where $k$ is a small constant, often between two and four.  
We consider three main versions of the problem, depending on the graph $G$.

\subparagraph*{(i) \texorpdfstring{$G$}{G} is a polygonal curve \texorpdfstring{$P$}{P} of length \texorpdfstring{$n$}{n}.} 
    This is the most basic version of the problem; we summarize the results in Figure~\ref{tbl:results}. We state running times for both decision and optimization algorithms, depending on the number $k$ of vertices in the query curve.   
    For $k=2$ (i.e., directed segments) and for $k=3$ (i.e., three-vertex curves), we construct data structures of size $O(n\log n)$, so that $\ddF(P,Q)$ can be computed in $O(\log^3 n)$ time, see Sections~\ref{sec:oracle-2}~and~\ref{sec:oracle-3}, respectively. In
    Section~\ref{sec:oracle-4}, 
    we describe a data structure of size $O^*(n)$ for $k=4$ (i.e., four-vertex curves), so that $\ddF(P,Q)$ can be computed in $O^*(n^{1/2})$ time. 
 In each of these cases, one can restrict the query to a vertex-to-vertex subcurve of $P$, specified at query time.

\begin{figure}
\centering
\begin{tabular}{ |c|c|c|c|c| } 
 \hline
size of $Q$ & \multicolumn{2}{c|}{decision problem} & \multicolumn{2}{c|}{optimization problem} \\ 
 $k$ & for $P$ & for subcurve of $P$ & for $P$ & for subcurve of $P$ \\ 
 \hline
 1 & $O(\log n)$ & $O(\log^2 n)$ & $O(\log n)$ & $O(\log^2 n)$\\ 
 2 & $O(\log^2 n)$ & $O(\log^2 n)$ & $O(\log^2 n)$ & $O(\log^2 n)$ \\ 
 3 & $O(\log^2 n)$ & $O(\log^2 n)$ & $O(\log^3 n)$ & $O(\log^3 n)$ \\
 4 & $O^*(n^{1/2})$ & $O^*(n^{1/2})$ &  $O^*(n^{1/2})$ &  $O^*(n^{1/2})$ \\
 \hline
\end{tabular}
\caption{Distance oracles for curves. Decision problem answers questions of the form: ``given $Q$ and $r$, is $\ddF(P,Q)\leq r$?'' Optimization problem computes the discrete Fr\'echet distance. We also offer a variant where at query time one can restrict the query to an arbitrary vertex-to-vertex subcurve of~$P$.}
\label{tbl:results}
\end{figure}

The case where the query curves are line segments was studied by Buchin et al.~\cite{BuchinHOSSS22} for the continuous (rather than discrete) Fr\'echet distance. They presented an $O(n\kappa^{3+\eps} + n^2)$-size data structure, where $\kappa \in [1,n]$ is a parameter set by the user and $\eps > 0$ is an arbitrarily small constant, such that given a query segment $ab$ one can compute the Fr\'echet distance between $ab$ and $P$ in $O((n/\kappa)\log^2 n)$ time (alternatively, between $ab$ and a point-to-point subcurve of $P$, specified at query time, in $O((n/\kappa)\log^2 n + \log^4 n)$ time). Thus, to achieve polylogarithmic query time, they need a data structure of size roughly $O(n^4)$, in contrast to $O(n\log n)$ for the discrete Fr\'echet distance (see Sections~\ref{sec:oracle-2}--\ref{sec:oracle-3}). 

It is not surprising 
that the bound on the size of the data structure that we obtain in the case of segment queries is much better than the bound of Buchin et al.~\cite{BuchinHOSSS22}. On the other hand, it is somewhat surprising that one can obtain a polylogarithmic bound (for queries of up to three vertices) and a sublinear bound (for four-vertex queries), using near-linear space (see Figure~\ref{tbl:results}).

Our results for the curve version are relatively technical and we defer them to Section~\ref{sec:curves}. They form the basis for the following, more general version.

\subparagraph*{(ii) \texorpdfstring{$G$}{G} is a tree \texorpdfstring{$T$}{T} with \texorpdfstring{$n$}{n} nodes.}
The main idea is to decompose $T$ into heavy paths~\cite{heavy} and use the aforementioned curve oracles. We show in Section~\ref{sec:trees} that one can construct discrete Fr\'echet distance oracles for a tree with $n$ vertices with query times $O(\log^3 n)$, $O(\log^3 n)$, $O(\log^4 n)$, and $O^*(\sqrt{n})$ for query sizes one, two, three, and four, respectively, where
the structures require $O(n\log n)$ space for query sizes up to three and $O^*(n)$ for query size four.   

\subparagraph*{(iii) 
\texorpdfstring{$G$}{G} is a local graph with \texorpdfstring{$n$}{n} vertices and \texorpdfstring{$m$}{m} edges.} 
Let $t \ge 1$ be a real parameter. We say that $G$ is \emph{$t$-local} if the following condition holds:
For any disk $D$ and for any two points $p,q \in P \cap D$ (where $P$ is the point set underlying $G$), there exists a path in $G$ between $p$ and $q$ that does not exit $tD$, where $tD$ is the disk obtained from $D$ by scaling it by a factor of $t$ around its center, that is, $p$ and $q$ are connected in the subgraph of $G$ that is induced by the set $P \cap tD$. 
We say that a graph is \emph{local} if it is $t$-local for some constant $t \ge 1$.

In Section~\ref{sec:graphs}, we first show that the class of local graphs \emph{strictly} contains the class of geometric spanners. Next, we show that any 1-local graph contains the Delaunay triangulation (which itself is 1-local). We construct an $O^*(n+m)$-size distance oracle for a given 1-local graph (and in particular an $O^*(n)$-size oracle for the Delaunay triangulation), which handles a segment query in $O^*(n^{1/2})$ time.
When $t > 1$, the oracle returns an approximation of the requested distance which depends on $t$.

\subparagraph*{More related work.}

We restrict our discussion of related work to distance oracles. In general, most of the related work deals with the continuous (rather than discrete) Fr\'echet distance, and with the construction of approximate oracles that return an approximation of the requested distance (rather than the exact distance). All the results below are for the continuous Fr\'echet distance unless mentioned otherwise.

As for exact oracles,
we already mentioned the result of Buchin et al.~\cite{BuchinHOSSS22} for arbitrary segment queries with respect to a given curve $P$. For earlier results geared to horizontal segment queries see~\cite{BergMO17,GudmundssonRSW21,BuchinHOSSS22}. 
Recently, Cheng and Huang~\cite{ChengH23} described a distance oracle for $k$-vertex query curves of size $O(kn)^{\poly(d,k)}$ than can process a query with respect to a point-to-point subcurve of $P$ (specified at query time) in time $O((dk)^{O(1)}\log(kn))$.

As for approximate oracles,
Filtser and Filtser~\cite{FiltserF21} construct a $(1 + \eps)$-approximate distance oracle for a given $n$-vertex curve $P$ and $k$-vertex query curves. Its size is $O(\frac{1}{\eps})^{kd}\log\frac{1}{\eps}$ 
and it computes a $(1+\eps)$-approximation of the discrete Fr\'echet distance between $P$ and a $k$-vertex query curve in ${O}^*(kd)$ time. %
Driemel and Har-Peled~\cite{DriemelH13} present a $(1+\eps)$-approximate distance oracle for segment queries (i.e., $k=2$) of size $O((\frac{1}{\eps})^{2d} \cdot \log^2 \frac{1}{\eps})$ and query time $O(d)$. 
They also consider the version in which the query is with respect to a point-to-point subcurve of $P$, specified at query time. For this version, the size of their data structure is $O(n (\frac{1}{\eps})^{2d} \cdot \log^2 \frac{1}{\eps})$ and the query time is $O(\eps^{-2} \log n \log \log n)$. Filtser~\cite{Filtser18} considered the latter version for the discrete Fr\'echet distance. By adapting techniques from Driemel and Har-Peled, she constructed a data structure of the same size and query time $O(\log n)$. 
Finally, for general $k$, Driemel and Har-Peled~\cite{DriemelH13} provide a constant-factor approximate distance oracle of size $O(nd \log n)$, which can answer distance queries between any
subcurve of $P$ and a $k$-curve query in $O(k^2 d \log n \log(k \log n))$ time.

A problem closely related to ours is the following.
Construct a compact data structure for a geometric graph $G$, such that given a query polygonal curve $Q$ of length $k$ one can quickly compute the minimum Fr\'echet distance between $Q$ and \emph{any} vertex-to-vertex path in~$G$. This is the query version of the well-known \emph{map matching} problem.
Gudmundsson and Smid~\cite{GudmundssonS15} studied the problem for a $c$-packed tree $T$. (A set of edges is $c$-packed if for any disk the total length of the portions of the edges contained in the disk is at most $c$ times the radius of the disk.) More precisely, they studied a corresponding decision problem with some additional restrictions.
Recently, Gudmundsson et al.~\cite{GudmundssonSW23} studied this problem for $c$-packed graphs. As an intermediate result, they construct a data structure of size $O(c\,m \log m)$ for a $c$-packed graph $G$ of complexity $m$, so that given a pair of query vertices $u$ and $v$, one can return in $O(\log m)$ time a 3-approximation of the Fr\'echet distance between $Q$ and $\P_G(u,v)$. The preprocessing time is $O(c\,m^2 \log^2 m)$.

\section{Distance oracles for trees}
\label{sec:trees}

A geometric tree is a tree whose vertices are points in the plane and whose edges are line segments connecting the corresponding points. In this section, we construct a discrete Fr\'echet distance oracle for a given geometric tree~$T$. In other words, we describe how to preprocess a tree~$T$ on $n$ vertices, so that, given a polygonal curve~$Q$ of size~$k$ and two vertices~$u$ and~$v$ of~$T$, one can efficiently compute the discrete Fr\'echet distance between~$Q$ and the path~$\Pi_{uv}$ in~$T$ from~$u$ to~$v$ (which is a polygonal curve), that is, one can quickly return $\dF(Q,\Pi_{uv})$.

As mentioned, we focus on the case where $k$ is a small constant. In this case, one can compute $\dF(Q,\Pi_{uv})$ without any preprocessing in linear time, so the goal is to do it in sublinear time after some preprocessing. More precisely, we only allow near-linear time preprocessing and storage. 

We present a reduction of our problem (discrete Fr\'echet distance oracle for trees) to that for polygonal curves. Specifically, assuming we already know how to construct a discrete Fr\'echet distance oracle for a polygonal curve and queries of size at most $k$, we construct a discrete Fr\'echet distance oracle for $T$ and queries of size $k$, at the cost of an additional logarithmic factor in the query bound. In Section~\ref{sec:curves}, we obtain discrete Fr\'echet distance oracles for polygonal curves that accept queries of size one, two, three, and four.
Thus, by our reduction, we immediately obtain the corresponding discrete Fr\'echet distance oracles for trees.

\subparagraph*{Black box: Distance oracle for curves.}
Fix $k\in\mathbb{N}$. 
Assume we have a black box that preprocesses a polygonal curve $P=(p_1,\ldots,p_n)$ in near-linear time such that, given a subcurve of $P$ between vertices $i$ and $j$, denoted $P[i,j]$ with $1 \le i \le j \le n$, and a query curve $Q$ of size at most $k$, it computes the discrete Fr\'echet distance between $Q$ and $P[i,j]$ in $t_k(n)$ time. We assume that $t_{k-1}(n) \le t_k(n)$.
Section~\ref{sec:curves} presents an implementation of the black box for query size up to four.

\subsection{Data structures}
\label{sec:tree-ds}

Let $T=(V,E)$ be a geometric tree on $n$ vertices. We first pick a root of $T$ arbitrarily and decompose $T$ into \emph{heavy paths}~\cite{heavy}.
The \emph{heavy-path decomposition} of a rooted tree~$T$ has the following properties: it is a collection of ``heavy paths;'' each heavy path is a (possibly degenerate) subpath of a leaf-to-root path in~$T$, beginning at a leaf; the top endpoint of each path (unless it is the root of~$T$) links to a node in another heavy path, in such a way that for any two vertices $u,v$ of $T$ the path between them in~$T$ switches between at most $O(\log n)$ heavy paths; every link in $T$ is either a heavy-path link or a link between the top node of a heavy path and its parent in~$T$. See Figure~\ref{fig:heavy}. %

\begin{figure}[ht] 
    \centering
    \includegraphics[scale=0.8]{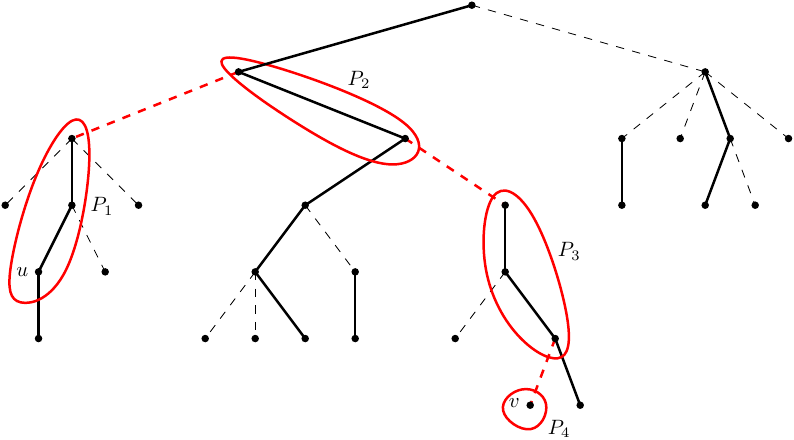}
    \caption{The heavy-path decomposition. The heavy paths are drawn in bold; the path $\Pi_{uv}$ is the concatenation of the subpaths $P_1,\ldots,P_4$.}
    \label{fig:heavy}
\end{figure}

Recall that, for vertices $u,v$ of the $T$, $\Pi_{uv}$ denotes the path from $u$ to $v$ in~$T$. Given $u,v$, one can compute the list of the $O(\log n)$ subpaths (of the decomposition's paths) whose concatenation is $\Pi_{uv}$ in $O(\log n)$ time, where each subpath is specified by the indices of its endpoints, as shown in \cite{heavy}. (The simpler data structure in \cite{heavy} supporting $O(\log^2 n)$-time query is sufficient for our purposes, as this is not the bottleneck in our approach.)

Next, for each path in the decomposition, we construct a discrete Fr\'echet distance oracle for polygonal curves for queries of size at most $k$.

\subsection{The Algorithm}
\label{sec:tree-alg}

Let $u,v\in V$ and let $Q=(q_1, q_2, \dots, q_k)$. We describe how to compute $\dF(\Pi_{uv},Q)$, using the data structures 
above.  
We begin by computing the representation of $\Pi_{uv}$ as the concatenation of $O(\log n)$ subpaths $P_1,P_2,\dotsc,P_m$ (i.e., $\Pi_{uv} = P_1 \cdot P_2 \cdot \ldots \cdot P_m$, where ``$\cdot$'' denotes concatenation of paths).
If $k=1$, then $Q=(q_1)$ and %
\begin{equation}
\dF(\Pi_{uv},(q_1))= \max\limits_{i\in\{1,\dotsc,m\}}\{\dF(P_i,(q_1))\} , 
\label{Tree base case}
\end{equation}
which can be computed in time $mt_1(n)=O(t_1(n)\log n)$.

For $k>1$, consider an optimal walk of $\Pi_{uv}$ and $Q$, and let $P_j$ be the last subpath to which $q_1$ is assigned, let $q_\ell$ be the last point of $Q$ that is assigned to $P_j$ (it is possible that $\ell = 1$), and let $q_{\ell'}$ be the first point of $Q$ that is assigned to $P_{j+1}$. Then $1 < \ell' \le k$ and $\ell' \in \{\ell,\ell+1\}$, and
        \begin{equation}
            \begin{aligned}
            \dF(\Pi_{uv},(q_1\dots q_k)) ={}  \max\{&\dF(P_1\cdot P_2\cdot\ldots\cdot P_{j-1},(q_1)), \\
                                        &\dF(P_j,(q_1,\dots,q_\ell)), \\
                                        &\dF(P_{j+1}\cdot\ldots\cdot P_m,(q_{\ell'},\dots,q_k)) \}.
        \end{aligned}
        \label{Tree case 2}
        \end{equation}

\subparagraph*{A recursive algorithm.} 
If $k=1$, return $\max\limits_{i\in\{1,\dotsc,m\}}\{\dF(P_i,(q_1))\}$, according to Eq.~\eqref{Tree base case}.
Otherwise, according to Eq.~\eqref{Tree case 2}, for each $j\in \{1,\dotsc,m\}$ and for each $\ell\in \{1,\dotsc,k\}$ and $\ell' \in\{\ell,\ell+1\}$, $1 < \ell' \le k$,
compute 
$\max\{\dF(P_1\cdot P_2\cdot\ldots\cdot P_{j-1},(q_1)),
\dF(P_j,(q_1,\dots,q_\ell)),
\dF(P_{j+1}\cdot\ldots\cdot P_m,(q_{\ell'},\dots,q_k)) \}$
(when $j+1>m$ or $j-1<1$, we ignore the relevant component) and return the smallest of all these values.

\subparagraph*{Dynamic programming procedure.} Our algorithm computes the values in Eqs.~\eqref{Tree base case} and~\eqref{Tree case 2} bottom up. For each subpath $P_j$ of $P$ and $Q[\ell_1,\ell_2]$ of $Q$, we calculate $\dF(P_j,(q_{\ell_1},\dots,q_{\ell_2}))$ using the curve oracle black box. Since $k$ is a constant, these calculations take $O(t_k(n)\log n+t_{k-1}(n)\log n+\dotsb+ t_1(n)\log n)=O(t_k(n)\log n)$, since we assumed $t_{k-1}(n)\le t_k(n)$. 

Then, for each $1\le j_1 < j_2\le m$ and for each $1\le \ell \le k$, we calculate $\dF(P_{j_1}\cdot\ldots\cdot P_{j_2},(q_\ell))$. Again, we calculate the values bottom up, starting from $j_2-j_1=1$ (computing $\dF(P_{j_1}\cdot\ldots\cdot P_{j_2+1},(q_\ell))$ takes $O(1)$ time if the answer to $\dF(P_{j_1}\cdot\ldots\cdot P_{j_2},(q_\ell))$ and $\dF(P_{j_2+1},(q_\ell))$ is known). We calculate $O(\log^2 n)$ values in $O(1)$ time each, so this step takes $O(\log^2 n)$ time in total.

Next, we calculate $\dF(P_{j_1}\cdot\ldots\cdot P_m,(q_{\ell_1},\dots, q_k))$ for each $1\le j_1 <  m$ and each $1\le\ell_1\le k$, starting from $j_1 = m-1$ and $\ell_1=k$, then $j_1 = m-1$
and $\ell_1=k-1$, etc.
Note that if all ``smaller'' values are already calculated, computing 
\begin{multline*}
    \dF(P_{j_1}\cdot\ldots\cdot P_m,(q_{\ell_1},\dots,q_k)) ={}\\
    \min\limits_{\substack{j\in\{j_1,\dots, m\}\\ \ell\in\{\ell_1,\dots,k\}\\ \ell'\in\{\ell, \ell+1\}\land \ell<\ell'\le k}} \left\{
    \begin{aligned}
    \max\{&\dF(P_{j_1}\cdot P_{j_1+1}\cdot\ldots\cdot P_{j-1},(q_{\ell_1})), \\
          &\dF(P_{j},(q_{\ell_1},\dots,q_{\ell})), \\
          &\dF(P_{j+1}\cdot\ldots\cdot P_m,(q_{\ell'},\dots,q_k))\}
    \end{aligned}
    \right\}\qquad
\end{multline*}
takes $O(\log n)$ time.
There are $O(\log n)$ such computations, 
so all of them together take $O(\log^2 n)$ time.
Hence, the total running time of the algorithm is $O(\log^2 n + t_k(n)\log n)$.

Using the results from Section~\ref{sec:curves} (see Figure~\ref{tbl:results}) for the black box implementation, we can therefore conclude with the following summary; for the construction time and space, preprocessing for individual heavy paths dominates the costs. 

\begin{theorem} \label{thm:reduction}
  For a geometric tree on $n$ vertices,
  one can construct discrete Fr\'echet distance oracles with query times $O(\log^3 n)$, $O(\log^3 n)$, $O(\log^4 n)$, and $O^*(\sqrt{n})$ for query sizes one, two, three, and four, respectively.
  The structures require $O(n\log n)$ space for query sizes up to three and $O^*(n)$ for query size four.
\end{theorem}

\section{Distance oracles for local graphs}
\label{sec:graphs}

A \emph{geometric graph} $G$ is a graph defined over a (finite) set $P$ of points in $\reals^d$ as vertices, and in which the weight of an edge $e=(p,q)$, $p,q \in P$, is the Euclidean distance $\|p-q\|$ between $p$ and $q$.
Let $t \ge 1$ be a real parameter. We say that $G$ is \emph{$t$-local} if the following condition holds:
For any ball $B$ and for any two points $p,q \in P \cap B$, there exists a path in $G$ between $p$ and $q$ that does not leave $tB$, where $tB$ is the ball obtained from $B$ by scaling it by a factor of $t$ around its center, that is, $p$ and $q$ are connected in the subgraph of $G$ induced by the set $P \cap tB$. 
We say that a graph is \emph{local} if it is $t$-local for some constant $t \ge 1$.

We first examine the connection between geometric spanners and local graphs. Recall that $G=G(P,E)$ is a \emph{$t$-spanner}, if for any any two points $p,q \in P$, there exists a path in $G$ between $p$ and $q$ of length at most $t\cdot \|p-q\|$ where the length of a path is the sum of the lengths of its edges. A \emph{spanner} is a graph that is a $t$-spanner for some constant $t \ge 1$.   

\begin{observation}
\label{obs:spanner_implies_local}
  Any geometric $t$-spanner is $2t$-local.\footnote{%
  We believe that the constant 2 can be improved with some additional effort.}
\end{observation}
\begin{proof}
Suppose $G=G(P,E)$ is a $t$-spanner.
We will show that $G$ is $2t$-local.

Consider an arbitrary pair of points $p,q \in P$ and put $d \coloneqq \|p-q\|$.  By assumption, there is a path $P(p,q)$ of length at most $td$ in $G$. Let $D$ be the disk with segment $pq$ as the diameter.  Let $x$ be a point of $P(p,q)$. By the triangle inequality, we have 
\begin{equation}  \|p-x\| + \|q-x\| \le |P(p,x)| + |P(q,x)| = |P(p,q)| \leq td\, , \label{ellipse} \end{equation}
where $P(p,x)$ and $P(q,x)$ are the appropriate subpaths of $P(p,q)$.
The locus of points $x$ satisfying Eq.~\eqref{ellipse} is an elliptical region $E$ with foci $p$ and $q$ and major axis $td$.  In particular, it fits into $tD$.  This proves the observation for disk $D$, as clearly $tD \subseteq 2tD$. 

Now, let $D'$ be any other disk containing both $p$ and $q$. We show that $tD \subseteq 2tD'$ and therefore $P(p,q) \subseteq 2tD'$. Let $o'$ and $d'$ be the center and diameter of $D'$, respectively. Then, $d' \ge d$ and $\|o'-p\|,\|o'-q\| \le d'/2$ (and therefore also $\|o'-o\| \le d'/2$, where $o$ is $D$'s center). Let $a$ be any point on the boundary of $tD$, then, by triangle inequality, $\|o'-a\| \le \|o'-o\| + \|o-a\| \le d'/2 + td/2 = d'(1/2 + t/2) \le td'$ (since $t \ge 1$), and therefore $a \in 2tD'$, completing the proof.
\end{proof}

The opposite implication does not hold, as formalized in Theorem~\ref{thm:local-not-spanner} below. %

We conclude that the locality property is weaker than the spanning property, i.e., the class of local graphs strictly contains the class of spanner graphs. 
\begin{theorem}
    \label{thm:local-not-spanner}
    There exists a constant $t > 1$, such that for any $t'\ge 1$, one can construct a graph that is $t$-local, but not a $t'$-spanner.
\end{theorem}

\begin{proof}
Consider the following sequence of path graphs. Let $F_0$ be the path with vertices at $(0,0)$ and $(1,0)$ and an edge between them. For $n \ge 1$, the path $F_n$ is obtained from the path $F_{n-1}$ by replacing each edge $(u,v)$ of $F_{n-1}$, where $u$ precedes $v$, with three new vertices and four new edges as follows. First, add vertices $u'$ and $v'$, so that $|uu'| = |u'v'| = |v'v| = \frac{1}{3}|uv|$, and add the edges $(u,u')$ and $(v',v)$. Next, add vertex $w$ to the left of the segment $u'v'$, so that $\Delta u'wv'$ is an equilateral triangle, and add the edges $(u'w)$ and $(w,v')$; see Figure~\ref{fig:fractal}.   

\begin{figure}[ht]
    \centering
    \includegraphics[width=0.5\linewidth]{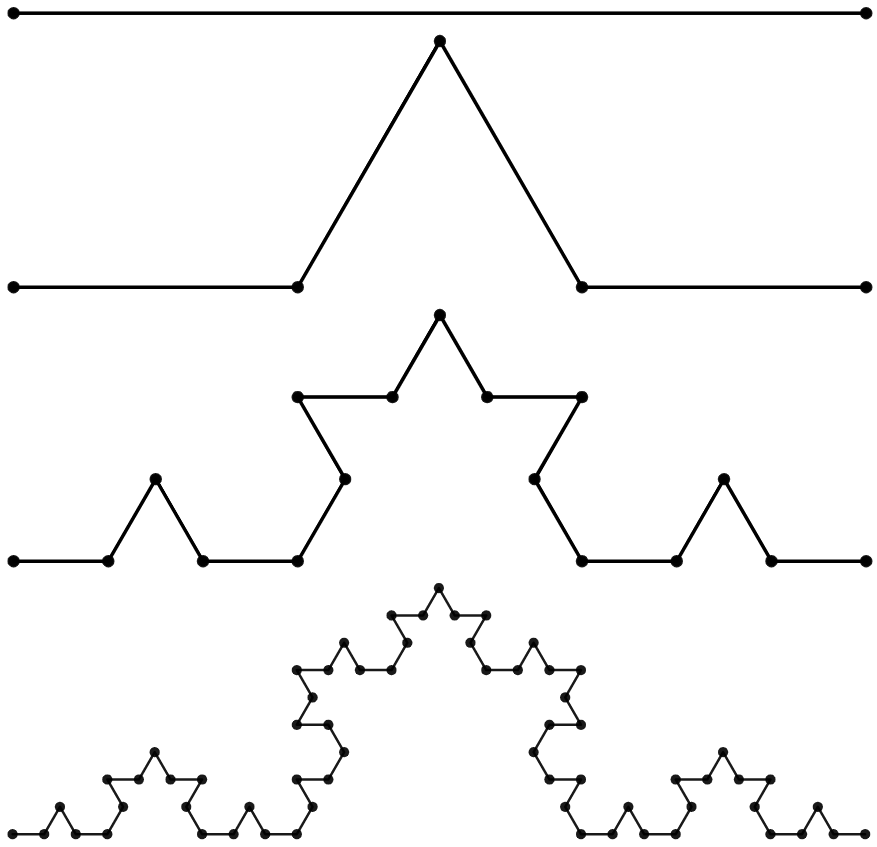}
    \caption{The path graphs $F_0$, $F_1$, $F_2$ and $F_3$.}
    \label{fig:fractal}
\end{figure}

By repeating this construction indefinitely, we get the fractal known as the \emph{Koch curve}~\cite{Koch1904}, which is one of the three curves forming the \emph{Koch snowflake}. It is well known that the length of the Koch curve is unbounded, that is, for every $l > 0$, there exists an integer $n$, such that the length of $F_n$, i.e., the sum of its edge lengths, is greater than $l$. Thus, for any $t' \ge 1$, there exists an integer $n$, such that $F_n$ is not a $t'$-spanner.

To complete the proof, we show that there exists a constant $t > 1$, such that for any $n \ge 0$, $F_n$ is $t$-local. 
Indeed, Farhana and Katz~\cite{FarhanaK23} showed that for any $n \ge 0$ and any two vertices $u$ and $v$ of $F_n$, there exists a rectangle $R$ such that (i) $R$ contains the subpath $F_n[u,v]$ of $F_n$ from $u$ to $v$, and (ii) the length of $R$'s diagonal is at most $c_0 \cdot \|u-v\|$, where $c_0$ is an absolute constant.  Now, consider a disk $D=\disk_r(o)$ containing $u$ and $v$. Then $r \ge \|u-v\|/2$.
Moreover, since for any vertex $w$ of $F_n[u,v]$, $\|u-w\| \le c_0 \cdot \|u-v\|$, we obtain
\[
  \|o-w\| \le \|o-u\| + \|u-w\| \le r + c_0 \cdot \|u-v\| \le (1 + 2c_0)r\, ,
\]
for any such vertex $w$. Or, in other words, $F_n[u,v] \subset (1+2c_0)D$, so $F_n$ is $t$-local for $t = (1+2c_0)$.
\end{proof}

Hereafter, we assume that $d=2$ and  that the points in $P$ are in \emph{general position}, i.e, no line passes through three or more points of $P$ and no circle passes through four or more of them.

\subparagraph*{1-local graphs.}
We begin with the case $t=1$, which is especially interesting.
Let $DT(P)$ denote the Delaunay triangulation of $P$. We think of $DT(P)$ as a graph over $P$, and prove below
that any 1-local graph over $P$ contains $DT(P)$ as a subgraph, and that $DT(P)$ itself is 1-local. 

\begin{observation}
\begin{enumerate}[(i)]
    \item
        $DT(P)$ is 1-local.
    \item
        Any 1-local graph $G$ over $P$ contains $DT(P)$.
\end{enumerate}
\end{observation}
\begin{proof}
    {\bf (i)} Let $D$ be a disk such that $|P \cap D| \ge 2$ and let $p, q \in P \cap D$. We need to show that there exists a path in $DT(P)$ between $p$ and $q$ that does not leave $D$, but this is a known property of $DT(P)$. 
    {\bf (ii)} Let $e=pq$ be an edge of $DT(P)$. Then, there exists a disk $D$ such that $P \cap D = \{p,q\}$.
    Since $G$ is 1-local, there exists a path in $G$ between $p$ and $q$ that does not leave $D$, so $e$ is an edge of $G$. 
    We thus conclude that $G$ contains $DT(P)$.
\end{proof}

We now return to our main topic, namely, discrete Fr{\' e}chet distance oracles, and study the following problem.
Let $G$ be a $t$-local graph defined over a set $P$ of $n$ points in the plane.
For any two vertices $u$ and $v$ of $G$, let $\P_G(u,v)$ denote the set of all paths between $u$ and $v$ in $G$.
Then, we define the \emph{discrete Fr{\' e}chet distance between a polygonal curve $Q$ and $G$} (with respect to $u$ and $v$) to be $\min_{\Pi \in \P_G(u,v)} \dF(Q,\Pi)$, denoted by $\dF(Q,\P_G(u,v))$.   
We wish to preprocess $G$, so that given a query curve $Q$ and two vertices $u$ and $v$ of $G$, one can compute $\dF(Q,\P_G(u,v))$ efficiently.

We begin with the special case where the queries are line segments connecting two vertices of $G$.

\subparagraph*{\texorpdfstring{$G$}{G} is 1-local and \texorpdfstring{$Q=uv$}{Q=uv}, where \texorpdfstring{$u,v$}{u,v} are vertices of \texorpdfstring{$G$}{G}.}
Let $r$ be the smallest radius for which there exist two vertices $u', v'$ of $G$, such that $u' \in \disk_r(u)$, $v' \in \disk_r(v)$, and $(u',v')$ is an edge of $G$, where $\disk_r(w)$ denotes the disk of radius $r$ centered at $w$. (Notice that if $(u,v)$ is an edge of $G$, then $r=0$.) Our solution is based on the following claim.

\begin{claim}
\label{cl:r}
	$\dF(uv,\P_G(u,v))=r$.
\end{claim}
\begin{proof}
    Let $\Pi = (u=w_1,\ldots,w_k=v) \in \P_G(u,v)$ be a path between $u$ and $v$ such that $\dF(uv,\P_G(u,v)) = \dF(uv,\Pi)$, and set $d^* = \dF(uv,\Pi)$. Let $\ell$, $1 \le \ell < k$, be a split index, i.e.,  if we associate $(w_1,\ldots,w_\ell)$ with $u$ and $(w_{\ell+1},\ldots,w_k)$ with $v$, then
    \[
      \max\bigl\{\max_{1 \le i \le \ell}\|u-w_i\|,\max_{\ell+1 \le i \le k}\|v-w_i\|\bigr\}=d^*\, .
    \]
    We first observe that $d^*$ is at least $r$. Indeed $w_\ell \in \disk_{d^*}(u)$, $w_{l+1} \in \disk_{d^*}(v)$, and $(w_l,w_{l+1})$ is an edge of $G$.
	
    To complete the proof we show that $r \ge d^*$. 
    For a contradiction, assume that $r < d^*$. Let $u'$ and $v'$ be two vertices such that $u' \in \disk_r(u)$, $v' \in \disk_r(v)$, and $(u',v')$ is an edge of $G$. Then, since $G$ is 1-local, there exists a path from $u$ to $u'$ contained in the disk around $u$ of radius $\|u-u'\| \le r$, and there exists a path from $v$ to $v'$ contained in the disk around $v$ of radius $\|v-v'\| \le r$. Therefore, there exists a path from $u$ to $v$ whose discrete Fr{\' e}chet distance from $uv$ is at most $r$ --- a contradiction.      
\end{proof}

\subparagraph*{The data structure.}
\label{sec:graph-ds}
In the preprocessing stage, we first construct a data structure $\Tedge$ of near-linear size for disk range searching in $P$ (see also the data structure description and references in 
Section~\ref{sec:oracle-4}). 
This data structure allows us to compute the set $P \cap D$, for a query disk $D$, as the union of pre-stored pairwise-disjoint canonical subsets, in $O^*(\sqrt{n})$ time. 
For a subset $P'$ of $P$, let $N(P')$ be the neighbor set of $P'$ in $G=G(P,E)$; that is, $N(P') = \{q \in P \mid \exists p \in P' \text{ with } (p,q) \in E\}$. We now augment $\Tedge$ as follows. For each canonical subset $P'$, we compute $V(P' \cup N(P'))$, the Voronoi diagram of $P' \cup N(P')$, and associate it with $P'$.
The final size of $\Tedge$ is therefore near-linear in $n+m$, where $m=|E|$. 

We also construct a second
data structure $\Tannu$ of near-linear size for annulus range searching in $P$. This data structure allows us to compute the set $P \cap A$, for a query annulus $A$, in $O^*(\sqrt{n} + k)$ time, where $k=|P \cap A|$.

\subparagraph*{The decision problem.} We describe how to determine whether $\dF(uv,\P_G(u,v)) \le d$, where $u$ and $v$ are any two vertices of $G$ and $d > 0$ is a given value, in $O^*(\sqrt{n})$ time. By arguments similar to those given above, $\dF(uv,\P_G(u,v)) \le d$ if and only if either $P \cap \disk_d(u) \cap \disk_d(v) \ne \emptyset$, or there exist points $u' \in \disk_d(u)$ and $v' \in \disk_d(v)$ such that $(u',v')$ is an edge of $G$. We thus use $\Tedge$ to compute a representation of $P \cap \disk_d(u)$ as the union of pairwise-disjoint canonical subsets, and for each of these subsets $P'$, we search in its associated Voronoi diagram $V(P' \cup N(P'))$ for the point that is closest to $v$. Finally, if at least one of the closest points that were found is within distance $d$ of $v$, then we return \textsc{yes}, and otherwise we return \textsc{no}.

\subparagraph*{Optimization.} We now describe how to compute $\dF(uv,\P_G(u,v))$, which is one of the $O(n)$ distances between $u$ or $v$ and a point in $P$. 
Let $S$ be a random sample of size $\sqrt{n}$ of these $O(n)$ distances. 
We find a pair of consecutive distances $d_1 < d_2$ in $S \cup \{0,\infty\}$, by a binary search using the decision procedure.
The expected number of distances between $u$ or $v$ and a point in $P$ that lie in the range $(d_1,d_2]$ is $O(\sqrt{n})$, and we can find them in $O^*(\sqrt{n})$ expected time by querying the second data structure $\Tannu$ with the annuli $A(u,d_1,d_2)$ and $A(v,d_1,d_2)$. Once we have these distances, we can find the smallest among them, $d^*$, that is still greater or equal than $\dF(uv,\P_G(u,v))$, by another binary search. We conclude that $\dF(uv,\P_G(u,v))=d^*$.

\subparagraph*{\texorpdfstring{$G$}{G} is 1-local and \texorpdfstring{$Q=ab$}{Q=ab}, where \texorpdfstring{$a,b$}{a,b} are arbitrary points in the plane.}
We remark that the query segment $Q$ does not have to be the segment between the specified vertices ($u$ and $v$) of $G$. The only difference is in the algorithm for the decision problem, where we need to take into account that $a$ must be `matched' to $u$ and $b$ must be `matched' to $v$. In particular, if $d < \max\{\|a-u\|,\|b-v\|\}$, then we immediately return \textsc{no}. Otherwise, we proceed as above, except that we consider the disk $\disk_d(a)$ (rather than $\disk_d(u)$) and search in the Voronoi diagrams for the point closest to $b$ (rather than to $v$).

The following theorem summarizes our result for $t=1$.
\begin{theorem}
    \label{thm:graph:1-local}
    Let $G=(P,E)$ be a 1-local graph. Then, we can compute $\dF(ab,\P_G(u,v))$, for any pair of vertices $u,v \in P$ and any pair of points $a,b \in \reals^2$, in $O^*(\sqrt{n})$ expected time, after a preprocessing stage in which we construct data structures of size $O^*(n+m)$. In particular, if $G$ is $DT(P)$, then the size of the data structures is $O^*(n)$. 
\end{theorem}

\subparagraph*{\texorpdfstring{$t$}{t}-local graphs, \texorpdfstring{$t > 1$}{t>1}.}
For $t > 1$, we use the same data structures and query algorithm to obtain an oracle that returns an approximation of the desired distance. More precisely, given two vertices $u$ and $v$ of $G$, the value $r$ returned by the query algorithm is such that $r \le \dF(uv,\P_G(u,v)) \le (t+1)r/2$. The proof is identical to the proof of Claim~\ref{cl:r}, except that now we only know that there exists a path from $u$ to $u'$ that is contained in the disk centered at $u$ of radius $(t+1)\|u-u'\|/2 \le (t+1)r/2$ and similarly for $v$ and $v'$. This follows from the $t$-locality property of $G$ applied to the disk of radius $r/2$ centered at the midpoint between $u$ and $u'$.  

As for the case $t=1$, we can also handle arbitrary segment queries. 
The following theorem summarizes our result for $t > 1$.
\begin{theorem}
    \label{thm:t-local-apx}
    Let $G=(P,E)$ be a $t$-local graph, $t > 1$. Then, for any pair of vertices $u,v \in P$ and any pair of points $a,b \in \reals^2$, we can compute a value $r$ such that $r \le \dF(uv,\P_G(u,v)) \le (t+1)r/2$ in $O^*(\sqrt{n})$ expected time, after $O^*(n+m)$ time and space preprocessing. 
\end{theorem}

By Observation~\ref{obs:spanner_implies_local} we obtain the following corollary.\begin{corollary}
    \label{cor:t-local-apx}
    Let $G=(P,E)$ be a $t$-spanner, $t > 1$. Then, given a query as above, we can compute a value $r$ such that $r \le \dF(uv,\P_G(u,v)) \le (2t+1)r/2$ in $O^*(\sqrt{n})$ expected time, after a preprocessing stage as above. 
\end{corollary}

\section{Black box revealed: Distance oracles for curves}
\label{sec:curves}

\subparagraph*{Notation and definitions.} 

Recall that we write $P[k,\ell]$, for $1 \le k \le \ell \le n$, to denote the (contiguous) subcurve $(p_k, p_{k+1}, \ldots, p_\ell)$ of $P$. For a point $q$ in the plane, the distance from $q$ to the vertex of $P[k,\ell]$ farthest from (nearest to) it, is denoted $d_{\max}(P[k,\ell], q)$ ($d_{\min}(P[k,\ell], q)$). 

Consider another curve $Q=(q_1,\dots,q_{\textrm{last}})$.
Put $\Delta=\Delta(P,Q)\coloneqq \max\{\|p_1-q_1\|,\|p_n-q_{\textrm{last}}\|\}$.  From the definition of a walk, it follows that $\ddF(P,Q)\geq \Delta(P,Q)$.

For two points $a$ and $b$ (which will usually be $q_1$ and $q_{\textrm{last}}$) and a real number $r \ge \max\{\|p_1 - a\|,\|p_n - b\|\}$, let $P_{\vdash}(r)$ denote the longest prefix of $P$, for which $d_{\max}(P_{\vdash}(r), a) \le r$, and let $P_{\dashv}(r)$ denote the longest suffix of $P$ for which $d_{\max}(P_{\dashv}(r), b) \le r$.

As a warm-up, we show how to solve the distance oracle problem for $k=1$.  In this case $\ddF(P,(a))=d_{\max}(P,a)$ and can be computed in logarithmic time by precomputing the farthest-neighbor Voronoi diagram of $P$ and preprocessing it for point-location queries.  The subcurve version of the problem (that is, computing $\ddF(P[k,\ell],(a))=d_{\max}(P[k,\ell],a)$) can be solved in $O(\log^2 n)$ time using the $\Tfvd$ structure defined below.  Clearly, solving the optimization problem answers the decision question within the same time bound.  This completes row $k=1$ in Figure~\ref{tbl:results}.

\subsection{\texorpdfstring{$k=2$}{k=2}}
\label{sec:oracle-2}

Let $P=(p_1,\ldots,p_n)$ be a sequence of points in the plane representing a polygonal curve. We construct a near-linear size data structure, that, given a 2-vertex query curve $Q=(a,b)$, can compute in $O(\log^2 n)$ time the discrete Fr\'echet distance $\ddF(P,Q)$ between $P$ and $Q$.

We begin with some definitions.  We say a distance $d$ satisfying $d\geq \ddF(P,(a,b))$ is \emph{feasible}.  As already observed, $d<\Delta\coloneqq\max(\|p_1-a\|,\|p_n-b\|)$ is not feasible. 
If $d \ge \Delta$ is a feasible distance with $d_{\max}(P_{\vdash}(d), a) = d$, we say 
that $d$ is \emph{prefix-feasible}.
Alternatively, if $d$ is feasible with $d_{\max}(P_{\dashv}(d), b) = d$, we say that it is \emph{suffix-feasible}.

\begin{observation}
  \label{obs:feasible-2}
  A distance $d\geq \Delta$ is feasible if and only if $P_{\vdash}(d)$ and $P_{\dashv}(d)$ cover $P$.
\end{observation}
\begin{proof}
  As already observed, any $d<\Delta$ is infeasible. A feasible $d$ corresponds to a walk that assigns a (non-empty) prefix of $P$ to~$a$ and a (non-empty) suffix to $b$, covering $P$.  The prefix is contained in $P_\vdash(d)$ and the suffix in $P_\dashv(d)$, completing one direction of the proof.

  Conversely, if $P_\vdash(d)$ and $P_\dashv(d)$ cover $P$ (neither can be empty, as $d\ge\Delta$), there is a (prefix,suffix) pair that can be assigned to $a$ and $b$ respectively, producing a walk of cost at most $d$, completing the proof.
\end{proof}

\subparagraph*{The data structure.}
We construct a binary tree, $\Tfvd$, which stores the farthest-neighbor Voronoi diagram~(FVD) of subsequences of $P$.  At the root of  $\Tfvd$, we store the FVD diagram of $P[1,n]$ together with a corresponding point location structure. In the left and right children of the root, we store the FVDs of $P[1,\lceil \frac{n}{2} \rceil]$ and $P[\lceil \frac{n}{2} \rceil + 1, n]$, respectively, etc. 
It is easy to see that the size of $\Tfvd$ is $O(n \log n)$, and that given a query point $q$ and indices $k,\ell$ ($1 \le k \le \ell \le n$), one can find the distance $d_{\max}(P[k,\ell],q)$ in $O(\log^2 n)$ time.
Moreover, given \emph{any} distance $d \ge \Delta$, one can compute $P_{\vdash}(d)$ and $P_{\dashv}(d)$ in $O(\log^2 n)$ time (roughly speaking, we descend $\Tfvd$ from the root checking canonical subsets for being within distance $d$ of $a$ or~$b$); we refer to this as a \emph{prefix} (resp., \emph{suffix}) computation.
Together with Observation~\ref{obs:feasible-2}, this gives an $O(\log^2 n)$ time test for feasibility, both for full $P$ and a subsequence of $P$.

\subparagraph*{Optimization.}
The algorithm consists of two symmetric parts, a left-to-right part and a right-to-left part.

The left-to-right part performs a binary search in $P$ to find the smallest prefix-feasible distance $d_L$.  We start by finding the distance $d_{\lceil \frac{n}{2} \rceil} = d_{\max}(P[1,\lceil \frac{n}{2} \rceil], a)$ (if $d_{\lceil \frac{n}{2} \rceil}<\Delta$, it is not prefix-feasible).  Next, we compute $P_{\vdash}(d_{\lceil \frac{n}{2} \rceil})$ and $P_{\dashv}(d_{\lceil \frac{n}{2} \rceil})$, by performing a prefix and a suffix computation using $\Tfvd$.  If together they cover $P$, then $d_{\lceil \frac{n}{2} \rceil}$ is prefix-feasible.  If so, we guessed too high.  If not, we guessed too low.  We continue with the binary search on half of the remaining sequence.

In the right-to-left-part, we perform a binary search in $P$ to find the smallest distance $d_R$ which is suffix-feasible. Finally, we output $\ddF(P,Q) = \min\{d_L,d_R\}$.  

\begin{lemma}
The algorithm above outputs the correct discrete Fr\'echet distance between $P$ and $Q$ in time $O(\log^3 n)$. 
\end{lemma}
\begin{proof}
Assume, without loss of generality, that $d^*= \ddF(P,Q)$ is determined by the distance between $a$ and $p_k$, for some $1 \le k \le n-1$. Then $d^*$ is prefix-feasible, and, since $d_L$ is the smallest distance which is prefix-feasible, we have $d_L \le d^*$. On the other hand, we have $d_L,d_R \ge d^*$. We conclude that the algorithm returns $d^*$ as claimed.

As for the running time, each iteration of the main binary search costs $O(\log^2 n)$ time. Thus, the total cost is $O(\log^3 n)$ time.
\end{proof}

With a little more care, we can improve the cost of each iteration of the main binary search to $O(\log n)$. The following theorem summarizes the main result of this section.

\begin{theorem}
  Given a curve $P=(p_1,\ldots,p_n)$ in the plane, one can construct a data structure of size $O(n \log n)$ such that, for any 2-vertex query curve $Q$, $\ddF(P,Q)$ can be computed in $O(\log^2 n)$ time. The same running time can be obtained with a subcurve of $P$ specified at query time.
\end{theorem}

\begin{proof}
A \emph{prefix-split} (resp., \emph{suffix-split}) index is an index $i$, $1 \le i \le n-1$, such that $d_{\max}(P[1,i], a) \ge d_{\max}(P[i+1,n], b)$ (resp., $d_{\max}(P[i+1,n], b) \ge d_{\max}(P[1,i], a)$) and $d_{\max}(P[1,i], a)-d_{\max}(P[i+1,n], b)$ (resp., $d_{\max}(P[i+1,n], b)-d_{\max}(P[1,i], a)$) is minimum.
The distance associated with a prefix-split (resp., suffix-split) index $i$ is $d_{\max}(P[1,i], a)$ (resp. $d_{\max}(P[i+1,n], b)$). We observe that it is sufficient to find a prefix-split index and a suffix-split index. Then, $\ddF(P,Q)$ is the smaller of their associated distances.
This follows from the monotonicity of the `distance from $a$' and `distance from $b$' functions: $d_{\max}(P[1,i], a) \le d_{\max}(P[1,i+1], a)$  and $d_{\max}(P[i+1,n], b) \le d_{\max}(P[i,n], b)$, for any $1 \le i \le n-1$.

We perform a binary search to find a prefix-split index (if it exists).
Assume, for simplicity, that $n$ is a power of 2. In the first step of the binary search, we set $s = \frac{n}{2}$ and use the tree $\Tfvd$ to compute $d_a \coloneqq d_{\max}(P[1,s],a)$ and $d_b \coloneqq d_{\max}(P[s+1,n]),b)$. If $d_a = d_b$, then $s$ is a prefix-split index (and also a suffix-split index) and we stop. If $d_a > d_b$, then we decrease $s$ (i.e., set $s = \frac{n}{4}$), unless $s=1$ or $d(P[s-1], a) < d(P[s], b)$, in which case $s$ is a prefix-split index and we stop. If $d_a < d_b$, then we increase $s$ (i.e., set $s = \frac{3n}{4}$), unless $s = n-1$, in which case there is no prefix-split index and we stop. We now repeat this step with the current value of $s$, etc. 

We show that with some additional care each step of the binary search can be implemented in $O(\log n)$ time. The idea is that whenever we compute a distance of the form $d_{\max}(P[1,i],a)$ (resp., $d_{\max}(P[i+1,n],b)$), we store it for future reference in the node of $\Tfvd$ that corresponds to the rightmost (resp., leftmost) canonical range in the partition of the range $[1,i]$ (resp., $[i+1,n]$) into a logarithmic number of canonical ranges. This way, whenever we need to compute such a distance, say, $d_{\max}(P[1,i],a)$, we only need to search in the farthest Voronoi diagram of a \emph{single} node of $\Tfvd$, namely, the one that corresponds to the rightmost canonical range in the partition of $[1,i]$ into canonical ranges, and then compute the maximum between (i) the distance between the reported vertex and $a$, and (ii) the distance that is stored in the previous canonical range in this partition.
For example, assume we have already computed $d_{\max}(P[1,3n/4],a)$ and now we wish to compute $d_{\max}(P[1,7n/8],a)$. Then, we only need to search in the Voronoi diagram of the node corresponding to the canonical range $[3n/4+1,7n/8]$ and use the distance that is stored in the node corresponding to the canonical range $[n/2+1,3n/4]$.

We now find a suffix-split index (if it exists), by performing a binary search in an analogous manner. Finally, we return the smaller of the distances associated with the two split indices that were found.  

For a subcurve $P[i,j]$, the overall cost remains $O(\log^2n)$. Roughly speaking, when searching for a prefix-split index, we first identify the subtree $T$ in which we need to perform a binary search, among the $k=O(\log n)$ canonical subtrees $T_1,\ldots,T_k$ representing the range $[i,j]$. Then, we perform a binary search in $T$ as described above. We can identify $T$ in $O(\log^2n)$ time by computing the distances $a_i = d_{\max}(a, T_{\ell})$ and $b_i = d_{\max}(b, T_{\ell})$, for $\ell = 1, \ldots, k$, where here $T_{\ell}$ stands for the subcurve of $P$ represented by $T_{\ell}$.
The search for a suffix-split index is done in a symmetric manner.

\end{proof}

\subsection{\texorpdfstring{$k=3$}{k=3}}
\label{sec:oracle-3}

Let $P=(p_1,\ldots,p_n)$ be a sequence of points in the plane representing a polygonal curve. We construct a near-linear size data structure that, given a 3-vertex query curve $Q=(a,b,c)$, computes in $O(\log^3 n)$ time the discrete Fr\'echet distance $\ddF(P,Q)$ between $P$ and $Q$.

\begin{claim} [Feasibility Test]
  \label{cl:feasible-3}
  For a distance $d \ge \Delta \coloneqq \max\{\|p_1 - a\|,\|p_n - c\|\}$, the following procedure decides if $d$ is \emph{feasible}, that is, if $d\geq \ddF(P,Q)$:
  \begin{quote}
    Let $i,j$ be the indices defined by $P_{\vdash}(d) = P[1,i]$ and $P_{\dashv}(d) = P[j,n]$. Now, (i) if $j > i+1$, then $d$ is feasible if $d_{\max}(P[i+1,j-1],b) \le d$, (ii) if $j=i+1$, then $d$ is feasible if $d_{\min}(P[i,j],b) \le d$, and (iii) if $j < i+1$, then $d$ is feasible if $d_{\min}(P[j,i],b) \le d$.
  \end{quote}
\end{claim}
Indeed, it is easy to verify that in each of the three cases, $d$ is feasible if and only if the appropriate condition holds.

We begin by adapting the definitions of prefix-feasible and suffix-feasible from Section~\ref{sec:oracle-2}.
We say that $d\geq \Delta$ is \emph{prefix-feasible} if $d$ is feasible and $d_{\max}(P_{\vdash}(d), a) = d$.  
Alternatively, it is \emph{suffix-feasible} if it is feasible and $d_{\max}(P_{\dashv}(d), c) = d$.

\subparagraph*{The data structure.}
We construct two binary trees, $\Tfvd$ and $\Tvd$.  The former has been described in Section~\ref{sec:oracle-2} and the latter is its analog for nearest-neighbor Voronoi diagrams. 

\begin{observation}\label{obs:feasible-3}
  The feasibility test takes $O(\log^2 n)$ time.
  In particular, given a distance $d \ge \Delta$, such that $d_{\max}(P_{\vdash}(d), a) = d$ (resp. $d_{\max}(P_{\dashv}(d), c) = d$), one can determine whether $d$ is prefix-feasible (resp., suffix-feasible) in $O(\log^2 n)$ time.
\end{observation}
\begin{proof}
Find in $O(\log^2 n)$ time the indices $i$ and $j$, such that $P_{\vdash}(d) = P[1,i]$ and $P_{\dashv}(d) = P[j,n]$, as in the previous section. Next, depending on whether $j > i+1$, $j = i+1$, or $j < i+1$, we verify the appropriate condition in $O(\log^2 n)$ time using $\Tfvd$ or $\Tvd$.
\end{proof}

\subparagraph*{Optimization.}
We assume for simplicity that all $3n$ distances are distinct and $\ddF(P,Q) > \Delta$.
(We can check whether $\ddF(P,Q) = \Delta$, by checking if $\|p_1-a\|$ is prefix-feasible or if $\|p_n-c\|$ is suffix-feasible, depending on which of the distances determines~$\Delta$.) 

The algorithm consists of two symmetric parts, a left-to-right part and a right-to-left part. Each part outputs a distance, and the smaller of these two distances is the desired distance, i.e., $\ddF(P,Q)$.

We describe the left-to-right part. 
We first perform a binary search to find the smallest distance $d_L$ which is prefix-feasible.  
Next, we find the largest distance ${\overline d_L}$ which is \emph{not} prefix-feasible. More precisely, let $p_L$ be the point of $P$ for which  $\|p_L-a\|=d_L$. Then ${\overline d_L}$ is determined by the point farthest from $a$ among the points of $P[1,L-1]$.
Clearly, ${\overline d_L} < d_L$ and $P_{\vdash}({\overline d_L})$ is strictly shorter than $P_{\vdash}(d_L)$. Next, we perform the process described in Claim~\ref{cl:feasible-3} above with the distance ${\overline d_L}$ to obtain the value $d'$. That is, let $1 \le i,j \le n$ be the indices such that $P_{\vdash}({\overline d_L}) = P[1,i]$ and $P_{\dashv}({\overline d_L}) = P[j,n]$. Now, (i) if $j > i+1$, then set $d' = d_{\max}(P[i+1,j-1],b)$, (ii) if $j=i+1$, then set $d' = d_{\min}(P[i,j],b)$, and (iii) if $j < i+1$, then set $d' = d_{\min}(P[j,i],b)$. Finally, the output of this part of the algorithm is $d_1 = \min\{d_L, d'\}$.

The output of the second part of the algorithm is $d_2 = \min\{d_R, d''\}$, where $d_R$ is the smallest distance which is suffix-feasible, and $d''$ is the value obtained by performing the analogous process with the distance ${\overline d_R}$. Given the outputs of both parts, we conclude that $\ddF(P,Q)=\min\{d_1,d_2\}$.

\begin{lemma}
The algorithm above is correct, i.e., it outputs the discrete Fr\'echet distance between $P$ and $Q$. Moreover, its running time is $O(\log^3 n)$. 
\end{lemma}
\begin{proof}
    Let $d^*$ denote the discrete Fr\'echet distance between $P$ and $Q$, i.e., $d^*=\ddF(P,Q)$. We distinguish between three cases, depending on which vertex of $Q$ defines $d^*$.  Cases I and II are symmetric and easy, while Case III is more involved.

\begin{description}
	\item[Case I: $d^*=\|p_k-a\|$.]
    This implies that $d^*$ is prefix-feasible. Moreover, it is clearly the smallest distance which is prefix-feasible, so $d^*$ will be found in the first part of the algorithm.
   	\item[Case II: $d^*=\|p_k-c\|$.] 
    The argument is entirely symmetric to Case~I.
    \item[Case III: $d^*=\|p_k-b\|$.]
    This implies that $d^* < d_L, d_R$. On the other hand, $d^* > {\overline d_L}, {\overline d_R}$, since ${\overline d_L}$ is not prefix-feasible and ${\overline d_R}$ is not suffix-feasible.
    Assume without loss of generality that ${\overline d_L} \ge {\overline d_R}$.
    We now claim that $P_{\vdash}({\overline d_L}) = P_{\vdash}(d^*)$. This is true, since $\overline{d}_L$ is the second largest distance among the distances between $a$ and the vertices of $P_{\vdash}(d_L)$ up to the vertex that determines $d_L$ and $d^* < d_L$.
    Similarly, we get that $P_{\dashv}({\overline d_R}) = P_{\dashv}(d^*)$, and therefore also $P_{\dashv}({\overline d_L}) = P_{\dashv}(d^*)$. This implies that the distance $d_1=d'$ returned by the left-to-right part of the algorithm is equal to $d^*$. \qedhere
\end{description}
\end{proof}

The following theorem summarizes the main result of this section.
\begin{theorem}
\label{thm:oracle-3}
Given a curve $P=(p_1,\ldots,p_n)$ in the plane, one can construct a data structure of size $O(n \log n)$ such that for any 3-vertex query curve $Q$, $\ddF(P,Q)$ can be computed in $O(\log^3 n)$ time.
The same running time can be obtained with a subcurve of $P$ specified at query time.
\end{theorem}

\subsection{\texorpdfstring{$k=4$}{k=4}}
\label{sec:oracle-4}

In this section we describe how to preprocess a piecewise-linear curve $P = (p_1,\ldots,p_n)$ so that given a piecewise-linear query curve $Q = (a,b,c,d)$, one can compute the discrete Fr\'echet distance between $Q$ and $P$.  In fact, we first describe how to solve the decision version of the problem and then use Matou\v{s}ek's randomized interpolating search \cite{Mat-search} to compute the actual distance.
Recall that $\disk_R(z)$ is the disk of radius $R$ centered at point $z$.

\subparagraph*{Preprocessing: Data structures.}

\begin{enumerate}
\item $\Tfvd$ and $\Tvd$: see Section \ref{sec:oracle-3}.

\item $\Tds$: Stores preprocessed edges of $P$, so that, for a query pair of indices $\ell,m$ and a distance $R$, one can determine whether there exists an edge $(p_k,p_{k+1})$ within $P[\ell,m]$ with $p_k \in \disk_R(b)$ and $p_{k+1} \in \disk_R(c)$. 

  We form a balanced binary search tree over the vertices of $P$, which, for a query subcurve $P[\ell,m]$, identifies canonical subsets of contiguous vertices of $P$ contained within $P[\ell,m-1]$.  
  For each subset, we build the $\Tedge$ data structure from Section~\ref{sec:graph-ds}.  We slightly modify $\Tedge$ in that a point $p_k$ is treated as only having one neighbor, namely $p_{k+1}$.  
  
  The resulting structure $\Tds$ is then able to support queries of the type: given $(\ell,m,b,c,R)$ is there an edge $(p_k,p_{k+1})$, $\ell \le k \le m-1$, with $p_k \in \disk_R(b)$ and $p_{k+1} \in \disk_R(c)$, or a point $p_k$, $\ell \le k \le m-1$, with $p_k \in \disk_R(b)\cap \disk_R(c)$? The data structure is built in time and space $O^*(n)$ with $O^*(n^{1/2})$ query time, using, e.g., the disk range searching data structure of Agarwal and Matou\v{s}ek~\cite{AgarwalM94}.

\item $\Tannu$: As in Section~\ref{sec:graphs},  this data structure stores the vertices of $P$ preprocessed for \emph{annulus reporting queries}. Given a point $q$ and radii $R_1,R_2$, return the set of all vertices in $P$ lying within the annulus centered at $q$ with inner radius $R_1$ and outer radius $R_2$. 
Here, the tools from \cite{arc-intersection-monster} yield  $O^*(n)$ expected space and construction time and $O^*(n^{1/2}+k)$ query time, where $k$ is the number of answers.

In the algorithm below we need to compute all inter-vertex distances between $P$ and $Q$ lying in the interval $[R_1,R_2]$. We perform four queries, one for the annulus centered at each of the vertices $a,b,c,d$ of $Q$.
\end{enumerate}

\subparagraph*{Decision procedure.}

We start with efficiently answering the question ``Is $\ddF(P,Q)\leq R$, for a given $R\ge\Delta \coloneqq \max\{\|p_1-a\|,\|p_n-d\|\}$?'' 
First we compute $P_{\vdash}(R) = P[1,i]$ and $P_{\dashv}(R) = P[j,n]$, as in Section~\ref{sec:oracle-2}, in $O(\log^2 n)$ time. 
The following two lemmas discuss how the relative positioning of $P[1,i]$ and $P[j,n]$ affects the answer to the decision problem.

\begin{lemma} \label{CaseI}

Suppose $j>i$. Then, $\ddF(P,Q)\leq R$ if and only if at least one of the following conditions holds:
\begin{enumerate}[(a)]
\item $\ddF(P[i+1,n], (b,c,d)) \leq R$
\item $\ddF(P[1,j-1],(a,b,c)) \leq R$
\item $j=i+1$ and $\ddF(p_ip_{i+1}, (b,c)) \leq R$
\end{enumerate}
\end{lemma}

\begin{proof}

In one direction, if $\ddF(P,Q)\leq R$, then there is a walk $W$ of $P$ and $Q$ of cost at most $R$. Consider the following two cases.

\begin{description}

\item[Case I: $j\geq i+2$.] Consider the subcurve $P[i+1,j-1]$ comprising the gap between $P_{\vdash}(R)$ and $P_{\dashv}(R)$. Clearly, in $W$, $P[i+1,j-1]$ must be matched entirely to vertex~$b$, entirely to vertex~$c$, or to the subcurve~$(b,c)$ of~$Q$. 
If $P[i+1,j-1]$ is matched entirely to~$b$, then $\ddF(P[i+1,n], (b,c,d)) \leq R$. If it is matched entirely to $c$, then $\ddF(P[1,j-1], (a,b,c)) \leq R$, and if it is matched to $(b,c)$, then both $\ddF(P[i+1,n], (b,c,d)) \leq R$ and $\ddF(P[1,j-1], (a,b,c)) \leq R$.
\item[Case II: $j=i+1$.]
We examine the pairs in $W$ involving $p_i$ and $p_j=p_{i+1}$: (i) if $(p_i,c) \in W$, then $\ddF(P[1,j-1], (a,b,c)) \leq R$, (ii) if $(p_{i+1},b) \in W$, then $\ddF(P[i+1,n], (b,c,d)) \leq R$, and (iii) if $(p_i,b) \in W$ and $(p_{i+1},c) \in W$, then $\ddF(p_ip_{i+1}, (b,c)) \leq R$. 
Since $W$ is a monotone walk of cost at most $R$, one of the above three cases must occur.
\end{description}

In the other direction, if $\ddF(P[i+1,n],(b,c,d)) \leq R$, then a walk of cost at most $R$ exists where $a$ is assigned to $P[1,i]$. If $\ddF(P[1,j-1],(a,b,c)) \leq R$, a symmetric walk can be constructed by assigning $d$ to $P[j,n]$. Finally, if $j=i+1$ and $\ddF(p_ip_{i+1}, (b,c)) \leq R$, we can compose a walk of cost at most $R$ where $a$ is assigned to $P[1,i]$, $b$ to $p_i$, $c$ to $p_{i+1}$, and $d$ to $P[i+1,n]$. 
\end{proof}

\begin{lemma} \label{caseII} Suppose $j \leq i$. Then, $\ddF(P,Q)\leq R$ if and only if one of the following holds:
\begin{enumerate}[(a)]
\item $\ddF(P[i+1,n], (b,c,d)) \leq R$
\item $\ddF(P[1,j-1],(a,b,c)) \leq R$
\item there exists a pair of consecutive vertices $p_k,p_{k+1}$ in $P[j-1,i+1]$ with $\ddF(p_kp_{k+1}, (b,c)) \leq R$ or a vertex $p_k$ in $P[j,i]$ with $\ddF(p_k, (b,c)) \leq R$. 
\end{enumerate}
\end{lemma}

\begin{proof} 
In one direction, suppose $\ddF(P,Q)\leq R$ and $j \leq i$. This implies there is a walk $W$ of $P$ and $Q$ of cost at most $R$. Now, if $(p_{i+1},b) \in W$, then $\ddF(P[i+1,n], (b,c,d)) \leq R$, and if $(p_{j-1},c) \in W$, then $\ddF(P[1,j-1],(a,b,c)) \leq R$. Otherwise, for any $k > i$, the pair $(p_k,b) \not \in W$, and for any $k < j$, the pair $(p_k,c) \not \in W$. We conclude that there exist two consecutive pairs in $W$, such that, either $(p_k,b),(p_{k+1},c) \in W$, where $j-1 \le k \le i$, thus satisfying the first case in condition (c), or $(p_k,b),(p_{k},c) \in W$, where $j \le k \le i$, thus satisfying the second case of condition (c). 

To argue the other direction, if $\ddF(P[i+1,n], (b,c,d)) \leq R$, then a walk of cost at most $R$ exists mapping $P[1,i]$ to $a$. Symmetrically, if $\ddF(P[1,j-1],(a,b,c)) \leq R$, assigning $P[j,n]$ to $d$ unveils a valid walk. If both (a) and (b) are false, suppose $P[j,i]$ has at least one vertex $p_k$ at most $R$ away from $b$ and $c$, or $P[j-1,i+1]$ has a pair of consecutive vertices $p_k,p_{k+1}$ with $\|p_k-b\|\le R$ and $\|p_{k+1}-c\|\le R$. We show there is a walk where $\ddF(P,Q) \leq R$. In the former case, we assign $a$ to $P[1,k]$, $p_k$ to $b$ and $c$, and $d$ to $P[k,n]$.
In the latter case, we assign $P[1,k]$ to $a$, $p_k$ to $b$, $p_{k+1}$ to $c$, and $P[k+1,n]$ to $d$.
\end{proof}

\subparagraph*{The decision algorithm.}

Check if $\ddF(P[i+1,n], (b,c,d)) \leq R$ or $\ddF(P[1,j-1],(a,b,c)) \leq R$, using the decision algorithm from Section~\ref{sec:oracle-3}.  If so, return \textsc{true}.  If both return \textsc{false} and $j>i$ then if $j=i+1$, run the decision algorithm from Section~\ref{sec:oracle-2} to check if $\ddF(p_ip_{i+1}, bc) \leq R$. If this query returns \textsc{true}, report \textsc{true}.
Otherwise, $j \leq i$ and we must check if there is an edge $p_kp_{k+1}$ within $P[j-1,i+1]$ such that $p_k$ lies in $\disk_R(b)$ and $p_{k+1}$ lies in $\disk_R(c)$ or a point $p_k$ in $P[j,i]$ no more than $R$ away from $b$ and $c$. We do this by querying $\Tds$ with $(j,i,b,c,R)$ in $O^*(n^{1/2})$ time, which performs all checks except for three. The remaining three are is any of the distances $\ddF((p_{j-1},p_j),(b,c))$, $\ddF((p_i,p_{i+1}),(b,c))$, and $\max(d(p_i,b),d(p_i,c))$ at most $R$, and we perform them separately in $O(1)$ time.
If such an edge or vertex is found, return \textsc{true}, otherwise return \textsc{false}.
See the pseudocode in Algorithm~\ref{alg:dec-4}.

\begin{algorithm}[ht]
\begin{algorithmic}
\State Compute $P[1,i]=P_{\vdash}(R)$ and $P[j,n]=P_{\dashv}(R)$ \Comment{Assuming $R\ge\Delta$}

\If{$\ddF(P[i+1,n], (b,c,d)) \leq R \lor
    \ddF(P[1,j-1],(a,b,c)) \leq R$}
    \State return \textsc{true}
\EndIf
\If{$j>i$} \Comment{Prefix and suffix do not overlap}
    \If{$j=i+1$}
        \State return {$\ddF(p_ip_{i+1}, bc) \leq R$}
    \EndIf
\Else  \Comment{Prefix and suffix overlap}
    \State Check $\Tds$ for $p_k$ in $P[j,i]$ in $\disk_R(b) \cap \disk_R(c)$ \Comment{See text description}
    \State ... or for $p_k,p_{k+1}$ within $P[j-1,i+1]$ with $p_k \in \disk_R(b)$ and $p_{k+1} \in \disk_R(c)$
    \If{either test returns \textsc{true}} 
        \State return \textsc{true}
    \Else                                      \State return \textsc{false}
   \EndIf
\EndIf 
\end{algorithmic}
\caption{Decision algorithm for $k=4$}
\label{alg:dec-4}
\end{algorithm}

\subparagraph*{Optimizaton via randomized interpolating search.}

We now address how to find the precise discrete Fr\'echet  distance between $P$ and $Q$. There are $4n$ critical values for the discrete Fr\'echet  distance, namely all the inter-point distances between the vertices of $P$ and $Q$. We cannot afford to list them all, so we use a variant of Matou\v{s}ek's randomized interpolating search, as in Section~\ref{sec:graphs}.

We produce a random sample $S$ of $\sqrt n$ Euclidean distances between pairs of points, one from $P$ and one from $Q$ (from the input lists).  Run the decision procedure for $R = \textrm{median}(S)$. Discard half of $S$ and repeat, until we find the precise discrete Fr\'echet distance 
or identify an interval $I$ of consecutive values of $S$ where the distance lies. We start with $\Delta\coloneqq\max(\|a-p_1\|, \|d-p_n\|)$ as the lower bound for the distance.

By a standard calculation, there are $O(\sqrt n)$ critical values within $I$ in expectation.
We use $\Tannu$ to extract these values in $O^*(n^{1/2})$ time, list the $O(\sqrt n)$ corresponding Euclidean distances, then binary search as above for the answer within this smaller list. 

\subparagraph*{A sublinear query.}
The correctness of our algorithm follows from the correctness of the decision procedure, as the rest is essentially a binary search on critical values. Since the bottleneck is the cost of $O(\log n)$ calls to the decision procedure, the query time is $O^*(n^{1/2})$.
We summarize the main result of this section.

\begin{theorem}\label{thm:oracle-4}
  Given a curve $P=(p_1,\ldots,p_n)$ in the plane, one can construct a data structure of expected size $O^*(n)$ such that for any 4-vertex query curve $Q$, $\ddF(P,Q)$ can be computed in $O^*(n^{1/2})$ time.
  The same running time can be obtained with a subcurve of $P$ specified at query time.
\end{theorem}

\old{
\begin{remark*}
  Alternatively, an efficient range-searching data structure that returns a random point in an annulus can be used to replace randomized interpolating search by a randomized binary search.
\end{remark*}
}

\bibliography{references}

\appendix
\end{document}